\documentclass[article,aps,prl,reprint,twocolumn,groupedaddress,10pt]{revtex4-2}
\usepackage{amsmath}
\usepackage{amsthm}
\usepackage{mathrsfs}
\usepackage{amsfonts}
\usepackage{amssymb}
\usepackage{graphicx}
\usepackage[colorlinks,linkcolor=blue,anchorcolor=blue,urlcolor=blue,citecolor=blue]{hyperref}
\usepackage{eufrak}
\usepackage{multirow}
\usepackage{textcomp}
\usepackage{epstopdf}
\usepackage{float}
\usepackage{epstopdf}

\begin{document}

\title{Experimentally demonstrating indefinite causal order algorithms to solve the generalized Deutsch's problem}

\author{Wen-Qiang Liu}\thanks{These two authors contributed equally to this work.}
\affiliation{Center for Quantum Technology Research and Key Laboratory of Advanced Optoelectronic Quantum Architecture and Measurements (MOE), School of Physics, Beijing Institute of Technology, Beijing 100081, China}	
\author{Zhe Meng}\thanks{These two authors contributed equally to this work.}
\affiliation{Center for Quantum Technology Research and Key Laboratory of Advanced Optoelectronic Quantum Architecture and Measurements (MOE), School of Physics, Beijing Institute of Technology, Beijing 100081, China}	
\author{Bo-Wen Song}
\affiliation{Center for Quantum Technology Research and Key Laboratory of Advanced Optoelectronic Quantum Architecture and Measurements (MOE), School of Physics, Beijing Institute of Technology, Beijing 100081, China}	
\author{Jian Li}
\affiliation{Center for Quantum Technology Research and Key Laboratory of Advanced Optoelectronic Quantum Architecture and Measurements (MOE), School of Physics, Beijing Institute of Technology, Beijing 100081, China}	
\author{Qing-Yuan Wu}
\affiliation{Center for Quantum Technology Research and Key Laboratory of Advanced Optoelectronic Quantum Architecture and Measurements (MOE), School of Physics, Beijing Institute of Technology, Beijing 100081, China}	
\author{Xiao-Xiao Chen}
\affiliation{Center for Quantum Technology Research and Key Laboratory of Advanced Optoelectronic Quantum Architecture and Measurements (MOE), School of Physics, Beijing Institute of Technology, Beijing 100081, China}	
\author{Jin-Yang Hong}
\affiliation{Center for Quantum Technology Research and Key Laboratory of Advanced Optoelectronic Quantum Architecture and Measurements (MOE), School of Physics, Beijing Institute of Technology, Beijing 100081, China}	
\author{An-Ning Zhang}
\email{Corresponding author. Anningzhang@bit.edu.cn}
\affiliation{Center for Quantum Technology Research and Key Laboratory of Advanced Optoelectronic Quantum Architecture and Measurements (MOE), School of Physics, Beijing Institute of Technology, Beijing 100081, China}	
\author{Zhang-qi Yin}
\email{Corresponding author. zqyin@bit.edu.cn}
\affiliation{Center for Quantum Technology Research and Key Laboratory of Advanced Optoelectronic Quantum Architecture and Measurements (MOE), School of Physics, Beijing Institute of Technology, Beijing 100081, China}

\date{\today }

\begin{abstract}
Deutsch's algorithm is the first quantum  algorithm to show the advantage over the classical algorithm. Here we generalize Deutsch's problem to $n$ functions and propose a new quantum algorithm with indefinite causal order to solve this problem.  The new algorithm not only reduces the number of queries to the black-box by half over the classical algorithm, but also significantly reduces the number of required quantum gates over the Deutsch's algorithm. We experimentally demonstrate the algorithm in a stable Sagnac loop interferometer with common path, which overcomes the obstacles of both phase instability and low fidelity of Mach-Zehnder interferometer. The experimental results have shown both an ultra-high and robust success probability  $\sim 99.7\%$.  Our work opens up a new path towards solving the practical problems with  indefinite casual order quantum circuits. 
\end{abstract}


\maketitle


\emph{Introduction.}---Deutsch's algorithm \cite{deutsch1985quantum} is the first quantum algorithm to show the quantum advantage. Later, some  well-known quantum algorithms, such as Shor's factorization algorithm \cite{Shor} and Grover's  search algorithm \cite{grover1997quantum}, also outperform  their classical peers and show the advantage of quantum  algorithms.  These algorithms can be implemented in quantum circuits with the fixed-gate order and exhibit well-defined causality in the geometry of spacetime.  However, the quantum circuit model is not a general framework to describe the quantum processes, because quantum mechanics allows the superposition of two or even multiple physical events in different temporal orders  \cite{oreshkov2012quantum,chiribella2013quantum}. One interesting problem is whether this general framework can be used to improve the quantum algorithms.

In the study of quantum gravity, researchers have discovered an exotic phenomenon in temporal orders called indefinite causal structure \cite{gambini2004relational,hardy2007towards,hardy2009quantum,christodoulou2019possibility}. Just like the superposition principle of quantum states, the superposition of  causal order leads to causal non-separability or quantum entanglement in time domain \cite{oreshkov2012quantum,araujo2015witnessing,oreshkov2016causal,rubino2017experimental,oreshkov2019time,cotler2017experimental,DONG20171235}, which is incompatible with operations that have the fixed order.  The use of indefinite causal structure to solve the task of quantum information process is a topic of great interest. As a new operational resource, the indefinite causal structures have shown potential advantages ranging from quantum computing \cite{araujo2014computational,procopio2015experimental,rambo2016functional,araujo2017quantum,renner2021reassessing,taddei2021computational,renner2022computational,escandon2023practical}, quantum communication \cite{chiribella2012perfect,guerin2016exponential,ebler2018enhanced,wei2019experimental,3-swithprocopio2019communication,guo2020experimental,procopio2020sending,goswami2020increasing,chiribella2021indefinite,rubino2021experimental,chiribella2021quantum}, quantum metrology \cite{zhao2020quantum,chapeau2021noisy,yin2023experimental}, to quantum thermodynamics \cite{guha2020thermodynamic,goldberg2021breaking,cao2022quantum,nie2022experimental,liu2022thermodynamics,simonov2022work,dieguez2023thermal}. For instance, it has been demonstrated that this new resource can remarkably reduce the query complexity in  quantum computing \cite{araujo2014computational,procopio2015experimental,renner2021reassessing,taddei2021computational,renner2022computational,escandon2023practical}, reduce the communication complexity \cite{guerin2016exponential,wei2019experimental} and improve the communication capacity \cite{ebler2018enhanced,3-swithprocopio2019communication,guo2020experimental,procopio2020sending,goswami2020increasing,chiribella2021indefinite,rubino2021experimental,chiribella2021quantum}. In addition, the indefinite causal order can also be used in quantum thermodynamics to improve the optimal heat-bath algorithmic cooling \cite{goldberg2021breaking} and enhance the thermodynamic efficiency \cite{cao2022quantum,nie2022experimental,liu2022thermodynamics}.

Quantum SWITCH \cite{chiribella2013quantum} is an efficient method for implementing the indefinite causal structure, where the operation order of two or more quantum gates is determined by control qubits. The simplest causal order superposition can be realized with a 2-SWITCH, which has been experimentally demonstrated to determine the commutation or anti-commutation of two unknown operations by querying each of them only once  \cite{procopio2015experimental}. Recently, both an $N$-SWITCH ($N>2$) \cite{3-swithprocopio2019communication,procopio2020sending} and a high-order quantum switch \cite{das2022quantum} have been proposed to  improve the efficiency of quantum information transmission. In addition, the $N$-SWITCH can  also be used to address the phase-estimation problem and the generalized Hadamard promise problem, and showed the computational advantage from indefinite causal structure over the fixed one \cite{taddei2021computational,renner2022computational,escandon2023practical}. Nevertheless, it still lacks the improvement of those well-known quantum algorithms with the indefinite causal order.

Up to now, the experimental realization of quantum SWITCH has mostly relied on the folded Mach-Zehnder interferometers (MZIs), where the path degree of freedom (DOF) of a single photon coherently  manipulates the order of two polarization operations. Unfortunately, the MZI based quantum SWITCH is faced with the challenges of both phase instability and low success possibility  \cite{procopio2015experimental,guo2020experimental,cao2022quantum,goswami2020experiments}. On one hand, it has to periodically adjust the phase to maintain the stability of the phase, which creates a lot of redundant work as the number of experiments increases. On the other hand, the geometric configuration of the folded MZI indicates that the photons on different arms undergo different polarization optical operations due to the non-common path modes of photons passing through the same device.  This would lead to additional errors and seriously degrade the fidelity of the experimental results. The MZI based quantum SWITCH so far can only be achieved with a success possibility around $97\%$ \cite{procopio2015experimental}.

In this Letter, we generalize the Deutsch's problem to the situation with $n$ Boolean functions, and propose a new algorithm with indefinite causal structure to solve the generalized  Deutsch's problem.  
Compared with the Deutsch's algorithm, the new algorithm reduces the gate number and has low circuit depth, therefore reveals the advantage of the indefinite casual order. Then, we experimentally demonstrate the new algorithm in a stable Sagnac loop geometric interferometer, which overcomes the obstacles of phase instability and low fidelity due to the perfect overlap of path mode of a photon and the full reciprocity of the polarized operations.
Compared with the previous experiments of the Deutsch's algorithm \cite{tame2007experimental,zhang2010demonstration,zhang2012implementing}, our experiment is greatly simplified.  The experimental results show  a robust and the currently maximum  success probability around $99.7 \%$, which goes far beyond results of the previous experiments \cite{tame2007experimental,zhang2010demonstration,zhang2012implementing}.


\emph{Theory.}---In 1985, Deutsch considered the following problem  whether a given Boolean function $f(x)$: \{0, 1\} $\mapsto$ \{0, 1\} is a balanced function or a constant function. Here, if the mapping result is $f(0)=f(1)=0$ or 1, one calls the function $f$ a constant function,  otherwise if $f(0)\neq f(1)$, one calls the function $f$ a balanced function. To solve the Deutsch's problem, it must query the function twice in a classical computer. Fortunately,
Deutsch found the so-called Deutsch's algorithm to solve the  problem only by querying the unknown function $f$ once in a quantum computer (see Fig. \ref{Deutsch}(a)). 

We here consider a generalized Deutsch's problem, which describes that there are $n$  Boolean functions $f_i(x)$ ($i$=1, 2, $\ldots, n$), each of them is either a balanced function or a constant function, and one wants to know whether there are an odd number of constant functions.  As shown in Fig. \ref{Deutsch}(b), we extend the Deutsch's algorithm to solve the generalized Deutsch's problem.  In this way, we require $n$ two-qubit black-box operations $U_{f_i}$ that embody in the unknown function $f_i(x)$, and  require to query each $f_i(x)$ one time (the details see Supplemental Material). In following text, we  will discuss how to solve the generalized Deutsch's problem by using indefinite causal order operations.

Figure \ref{Deutsch}(c) schematically illustrates a $2$-SWITCH with indefinite causal order, which coherently controls the operation order: gate $U_2$ is before $U_1$ if the controlled qubit is in the state $|0\rangle_c$, and gate $U_1$ is before $U_2$ if the controlled qubit is in the state $|1\rangle_c$. If the initial controlled qubit is prepared in the state $|+\rangle_c=\frac{1}{\sqrt{2}}(|0\rangle_c+|1\rangle_c)$ (see Fig. \ref{Deutsch}(d)), after a 2-SWITCH applies on $U_1$ and $U_2$, and then a Hardmard gate acts on the controlled qubit, the state of whole system becomes
\begin{eqnarray}              \label{eq3}
\begin{aligned}
\frac{1}{2}\big(|0\rangle_c\otimes \{U_1,  U_2\}|\psi\rangle_t+|1\rangle_c\otimes [U_1, U_2]|\psi\rangle_t\big).
\end{aligned}
\end{eqnarray}
Here $|\psi\rangle_t$ is the arbitrary target state operated by the gates $U_1$ and $U_2$. [$U_1$, $U_2$]=$U_1U_2$-$U_2U_1$ and \{$U_1$, $U_2$\}=$U_1U_2$+$U_2U_1$ denote the commutator and anti-commutator of operators $U_1$ and $U_2$, respectively. The commutation of the operators $U_1$ and $U_2$ can be determined from the measurement results of the controlled qubits in Eq. (\ref{eq3}). If the result is $0$ ($1$), then  
$U_1$ and $U_2$ are commutative (anti-commutative).

The 2-SWITCH can be used to solve the Deutsch's and the generalized Deutsch's problems. 
Here we define an operator $D(f_i)$=$\sum\limits_{x\in\{0, 1\}}(-1)^{f_i(x)}|x\rangle\langle x|$ ($i$=1, 2, $\ldots$, $n$), a gate $U_1=D(f_1)D(f_2)\ldots D(f_n)$ and a gate $U_2=X$  (a Pauli $X$ gate). For any input state $|y\rangle \in \big\{|0\rangle, |1\rangle\big\}$, (i) if $n$ is odd and $[U_1, U_2]|y\rangle$=0, then there are an odd number of constant functions; otherwise if $\{U_1, U_2\}|y\rangle$=0, then the number of constant functions is not odd. (ii) if $n$ is even and $[U_1, U_2]|y\rangle$=0, then the number of constant functions is not odd; otherwise if $\{U_1, U_2\}|y\rangle$=0, then there are an odd number of constant functions  (the proof see Supplemental Material). 

\begin{figure}   
\begin{center}
\includegraphics[width=8.2 cm,angle=0]{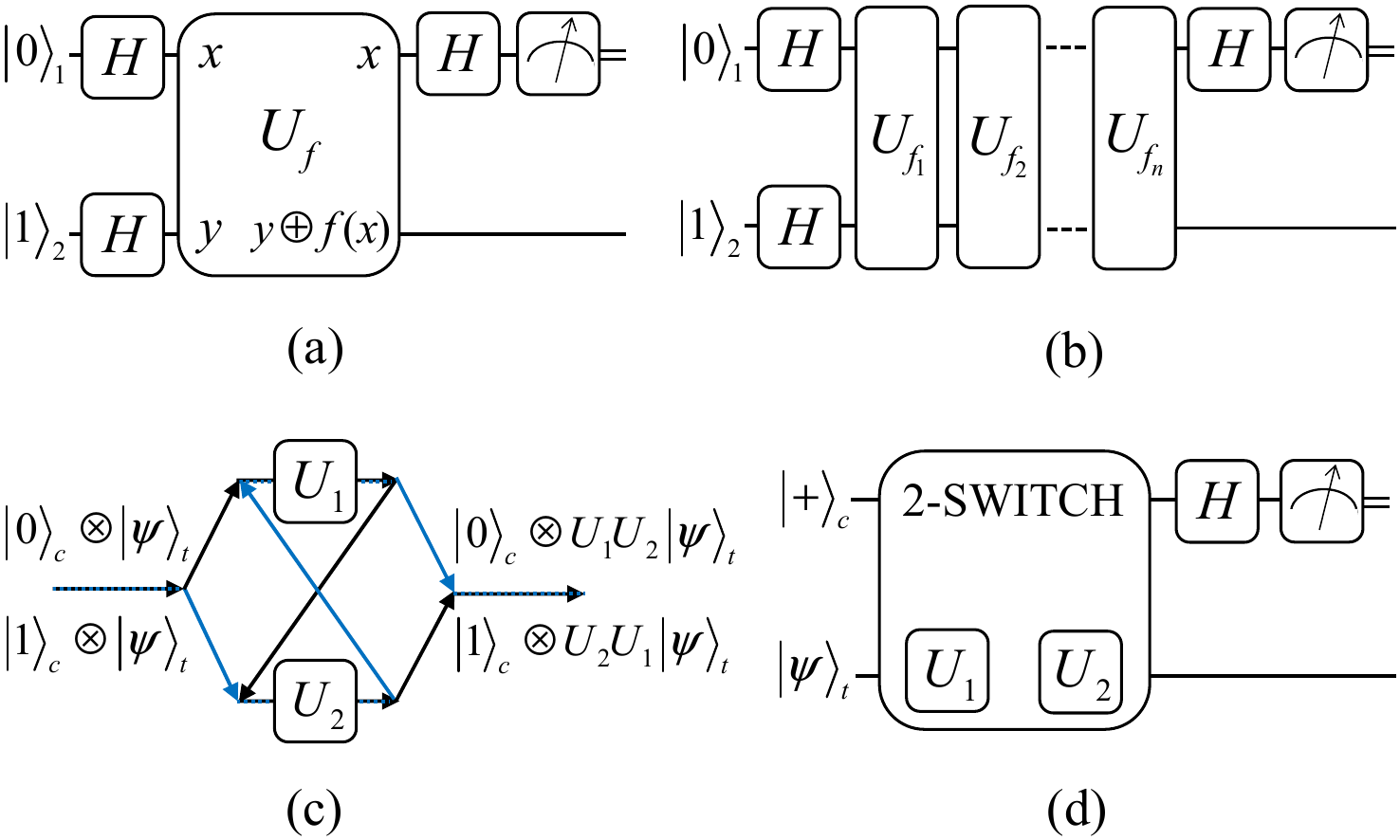}
\caption{(a) A well-known Deutsch's algorithm. $H$ is a Hardmard gate and a black-box operation $U_f$ realizes the transformation $U_{f}|x\rangle|y\rangle=|x\rangle|y\oplus f(x)\rangle$, where $\oplus$ denotes addition modulo 2. (b) A generalized Deutsch's algorithm we proposed. Each black-box operation $U_{f_i}$ realizes $U_{f_i}|x\rangle|y\rangle=|x\rangle|y\oplus f_i(x)\rangle$. (c) A quantum 2-SWITCH, which allows a coherent control of the operation order: when the controlled qubit is $|0\rangle_c$, the operation order is $U_2$ before $U_1$ (blue circuit), and when the controlled qubit is $|1\rangle_c$, the operation order is $U_1$ before $U_2$ (black circuit). (d) Our new algorithm to solve the generalized Deutsch's problem using the 2-SWITHCH. $|+\rangle_c=\frac{1}{\sqrt{2}}(|0\rangle_c+|1\rangle_c)$ and $|\psi\rangle_t$ is an arbitrary target state. Operation $U_1$ denotes the product of $n$ matrices that derive from $\pm I$ and $\pm Z$, which is determined by the value of each function $f_i(x)$. Operation $U_2$ denotes a Pauli $X$ gate. Here $I$ and $Z$ are a $2\times2$ identity matrix and a Pauli $Z$ gate, respectively.} \label{Deutsch}
\end{center}
\end{figure}

From the definition of $D(f_i)$, one can see that $D(f_i)$ is a diagonal matrix with entries $\pm1$ in the computational basis.  The operator $U_1$ contains the information of all functions $f_i(x)$. Specifically, $U_1$ is the product of $n$ matrices that derive from $\pm I$ and $\pm Z$, which is determined by the value of each function $f_i(x)$. Here $I$ and $Z$ are $2\times2$ identity matrix and Pauli $Z$ matrix, respectively. We compare our algorithm and the classical algorithm and the generalized Deutsch's algorithm to solve the generalized Deutsch's problem. For the classical algorithm, one has to query each $f_i(x)$ twice, while our algorithm and generalized Deutsch's algorithm only need to query each $f_i(x)$ once. In the generalized Deutsch's algorithm, the type of each function is determined by a two-qubit black-box gate $U_{f_i}$, which requires a series of the complex two-qubit controlled-NOT (CNOT) gates to realize. In contrast, in our algorithm, the type of each function is determined by a one-qubit black-box gate $D(f_i)$. A controlled qubit controls the overall  order of the black-box gates, which has much lower circuit depth than the generalized Deutsch's algorithm. As the complex CNOT gate is not required, our algorithm can be realized with very good scalability.

\begin{figure}   
\begin{center}
\includegraphics[width=8.3 cm,angle=0]{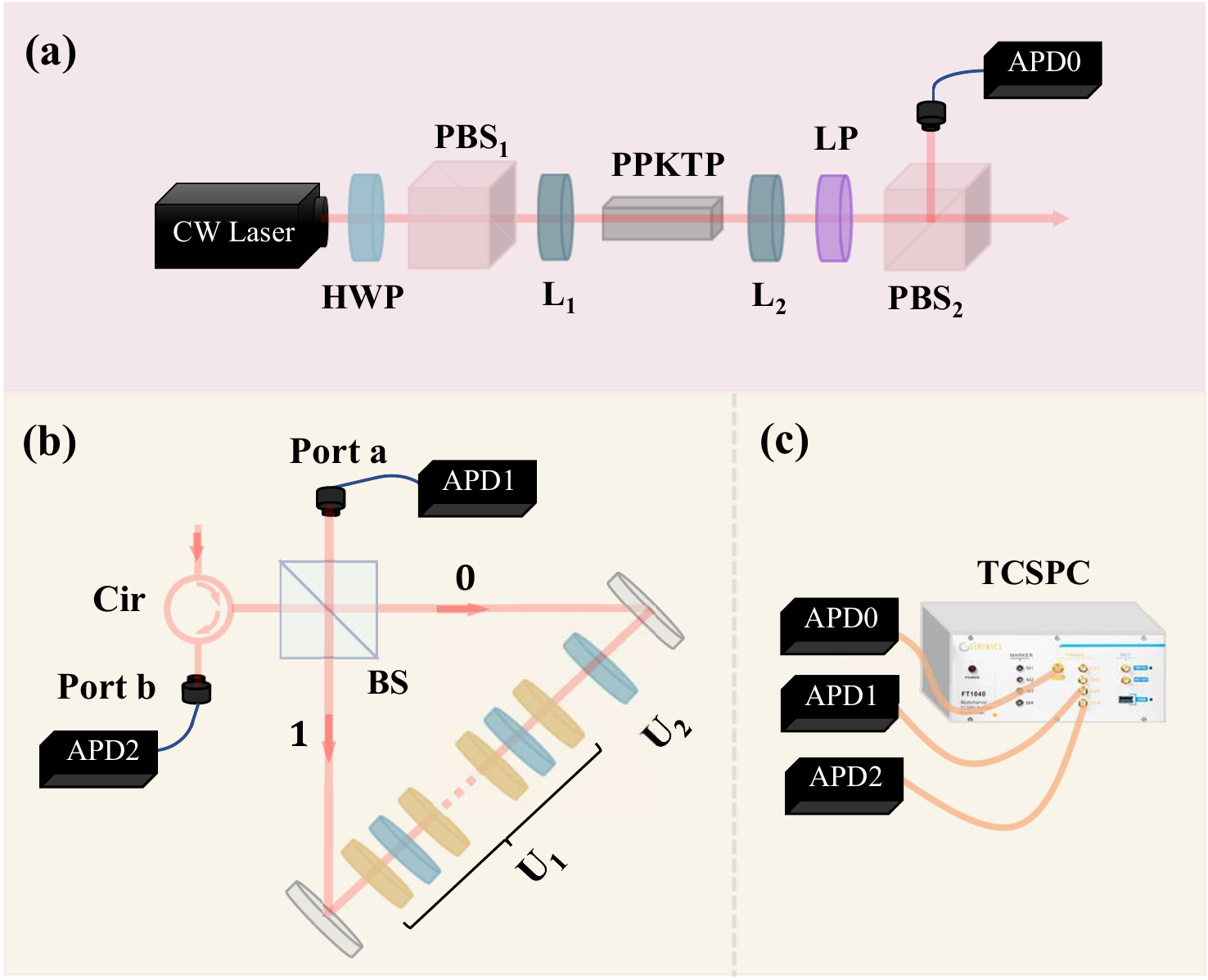}
    \caption{Schematic of the experimental setup to solve the generalized Deutsch's problem in a Sagnac geometric interferometer with indefinite causal order. (a) The preparation of a heralded single-photon source.  The single photons are created via a type-II spontaneous parametric down-conversion (SPDC) using a
poled potassium titanyl phosphate (PPKTP) crystal. (b) The implementation of a Sagnac loop interferometer with indefinite causal order. $U_1$ and $U_2$, two gate operations, which can be realized by half-wave plates  (HWPs) and quarter-wave plates (QWPs). (c) A counting module. PBS, polarizing beam splitter; L$_1$ and L$_2$, lenses; LP, long pass filter; Cir, circulator; BS, 50:50 beam splitter;  APD, avalanche photon-diode; TCSPC, time-correlated single-photon counting.} \label{setup}
\end{center}
\end{figure}

\begin{figure}   [htpb] 
\begin{center}
\includegraphics[width=8.4 cm,angle=0]{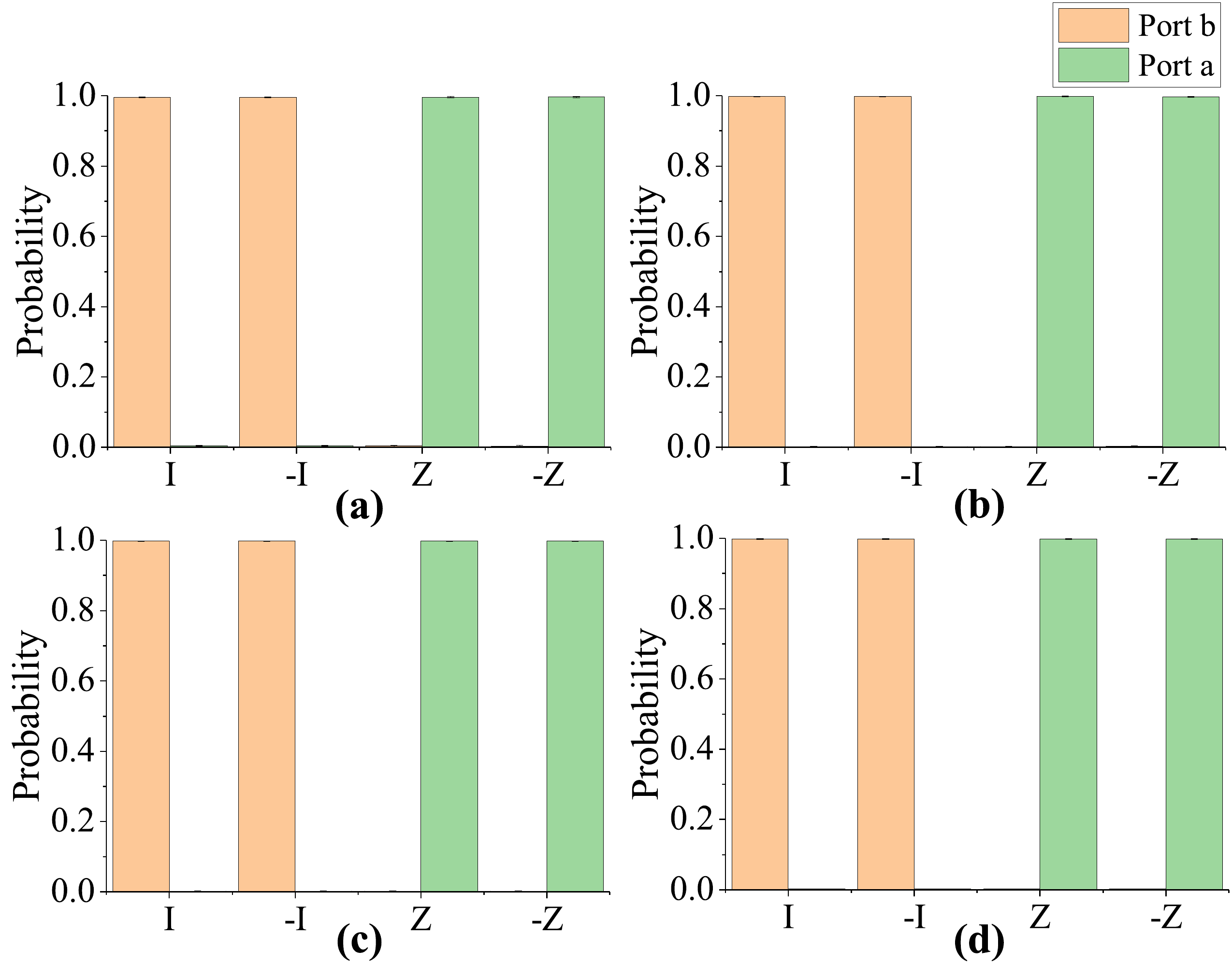}
    \caption{Experimental results to solve the Deutsch's problem. (a)-(d) show the experimental probabilities of the photons exiting from port $a$ or port $b$ in the polarization states $|H\rangle$, $|V\rangle$, $|D\rangle$, and $|A\rangle$, respectively, when determining whether $f$ is a balanced function or a constant function. The orange bar represents the observed probability of the photon exiting from port $b$, and the green bar represents the probability of the photon exiting from port $a$. If $f$ is a constant function, the photons exit from port $b$ ideally, while if $f$ is a balanced function, the photons exit from port $a$ ideally. The x-axis indicates the selection of the gate $U_1$, and $U_2$ is always Pauli $X$ gate. The average success probability of these data is 0.9972$\pm$0.0013.} \label{F1}
\end{center}
\end{figure}

\emph{Experiments.}---Based on the folded MZI, Procopio \emph{et al.} have experimentally demonstrated a 2-SWITCH to determine whether two gate operations are commutative or anti-commutative \cite{procopio2015experimental}. However, this configuration suffers from the phase instability. Recently, an optimized quantum SWITCH was realized by using Sagnac geometric interferometer \cite{stromberg2022demonstration}. Motivated by this, we experimentally realize our algorithm in the stable Sagnac configuration, where the path DOF of a single photon as the controlled qubit and the polarization DOF of the single photon as the target qubit. This implies that the path DOF is used for the coherent superposition of two different gate orders acting on the polarization DOF.

The experimental setup to demonstrate our algorithm for solving the generalized Deutsch's problem is illustrated in Fig. \ref{setup}, which incorporates a heralded single-photon source, a Sagnac loop interferometer, and a counting module. As shown in Fig. \ref{setup}(a), a continuous-wave diode laser emits a pump laser with a central wavelength of 405 nm and a power of 20 mW. The pump laser is used to generate photon pairs  with a wavelength of 810 nm via type-II spontaneous parametric down-conversion (SPDC) in a periodically poled potassium titanyl phosphate  (PPKTP) crystal.  We utilize a half-wave plate (HWP) and a polarization beam splitter (PBS$_1$) to regulate optical power, and use two lenses (L$_1$ and L$_2$) before and after the PPKTP crystal to focus and collimate beams. Subsequently, the photon pairs are filtered with a long pass filter (LP) to remove the pumped laser, and then are split on PBS$_2$ before the photons couple into a single-mode fiber. One photon of each photon pair is detected to herald the presence of the idler photon, and the other photon is used to inject into a Sagnac interferometer.

As shown in Fig. \ref{setup}(b), a signal horizontally polarized photon $|H\rangle$ of each photon pair is injected into the Sagnac interferometer from the upper entrance. The signal photon first is directed to a 50:50 beam splitter (BS) by a circulator (Cir). The BS acts on the path DOF of the photon, and realizes the transformation $|H\rangle \stackrel{\mathrm{BS}}{\longrightarrow} \frac{1}{\sqrt{2}}\left(|H\rangle_0+|H\rangle_1\right)$. Here subscripts 0 and 1 denote the paths of propagating photon. Subsequently, the photon in path 0 travels clockwise through the gate $U_2$ and then $U_1$, while the photon in path 1 travels counterclockwise through the gate $U_1$ and then $U_2$.  Finally, the photon in two paths  coherently recombine at BS and obtain an output state
%
\begin{eqnarray}              \label{eq6}
\begin{aligned}
\frac{1}{2}\{U_1, U_2\}|H\rangle_a+\frac{1}{2}[U_1, U_2]|H\rangle_b.
\end{aligned}
\end{eqnarray}
We can determine that whether there are an odd number of constant functions  by measuring the state of the path qubit (i.e., the photon exits from port $a$ or part $b$ of the BS). We note that arbitrary  polarized one-qubit gate can be realized by using two quarter-wave plates (QWPs) and one HWP \cite{simon1989universal,simon1990minimal,simon2012hamilton}. Here $D(f_i)$ and $U_2$ in our algorithm are some Pauli gates, whose realizations are presented in Tab. \ref{table3} of Supplemental Material. In experiments, the black-box operation $U_1$ is constructed by the setting of a series of $D(f_i)$.

To acquire the experimental data, we first initialize the $U_1$ gate as the identity matrix $I$ to adjust the phase of the interferometer to zero via  utilizing a liquid crystal (not depicted in Fig. \ref{setup}).  The all possibilities for each black-box function $f_i(x)$ embodied by gate $U_1$, which can be realized by manipulating the wave-plate angles. Note that all $U_1$ gates have good reciprocity in this Sagnac loop interferometer, that is, photons travelling both clockwise and counterclockwise through the wave plates can correctly realize $U_1$. While for the realization of the gate $U_2$, it  introduces a relative phase $\pi$ in this interferometer when photons pass through the wave plates clockwise and counterclockwise. This leads to a conclusion contrary to Eq. (\ref{eq6}):  the photons exiting from port $a$  mean that $U_1$ and $U_2$ are anti-commutative, while the photons exiting from port $b$ mean that $U_1$ and $U_2$ are commutative.

\begin{figure}    [htbp]
\centering
\includegraphics[width=8.4 cm,angle=0]{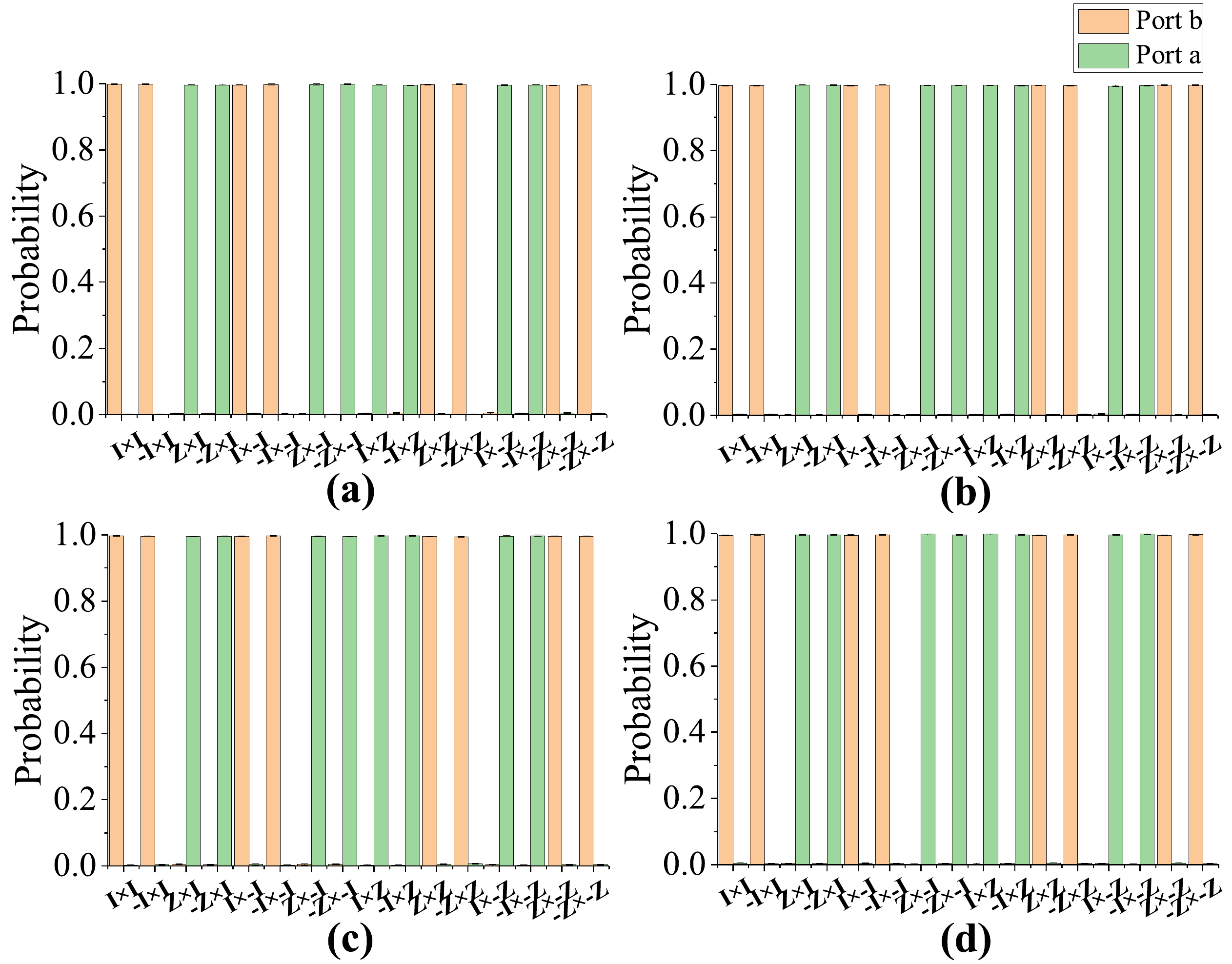}
    \caption{Experimental results to solve the generalized Deutsch's problem with two unknown functions $f_1$ and $f_2$. (a)-(d) show the experimental probabilities of the photons exiting from port $a$ or port $b$ in the polarization states $|H\rangle$, $|V\rangle$, $|D\rangle$, and $|A\rangle$, respectively, when determining whether there is a constant function in functions $f_1$ and $f_2$. The x-axis indicates the selection of the gate $U_1=D(f_1)\times D(f_2)$. The average success probability of these data is 0.9968$\pm$0.0014.} \label{F1F2}
\end{figure}

We first demonstrate our algorithm to solve the Deutsch's problem and a correspondence is presented in Tab. \ref{table1} of Supplemental Material.  We collect 600 thousand  experimental outcomes for the injected $H$-polarized photon by monitoring the responses from the port $a$ or the port $b$.  We evaluate the probability of photons to exit at each port, and plot the results in Fig. \ref{F1}(a). When the $U_1=\pm I$ (a constant function and commutes with $U_2=X$), we expect all photons to exit at port $b$ ideally, while when the $U_1=\pm Z$ (a balanced function and anti-commutes with $U_2=X$), we expect all photons to exit at port $a$ ideally. As shown in Fig. \ref{F1}, the experimental results agree with theoretical expectations.  To demonstrate that our algorithm does not depend on the initial incident polarized photon, we experimentally demonstrate the performance of our algorithm in different incident polarized basis \{$|H\rangle$, $|V\rangle$\}  and \{$|D\rangle$, $|A\rangle$\}. Here $|H\rangle$, $|V\rangle$, $|D\rangle$, and $|A\rangle$  correspond to horizontally, vertically, diagonally, and anti-diagonally polarized photon, respectively. We find that the minimum success probability of photons to exit at the expected port is 0.9950, and the average success probability is 0.9972$\pm$0.0013, which far exceeds the previous results \cite{tame2007experimental,zhang2010demonstration,zhang2012implementing}.

We also experimentally demonstrate our algorithm to solve the generalized Deutsch's problem with two functions (see Tab. \ref{table4}  in Supplemental Material) and the experimental results are plotted in Fig. \ref{F1F2}. The results show that the average success probability of photons to exit at the expected port achieves 0.9968$\pm$0.0014, which is good agreement with the theoretical expectations. The error bars correspond to the $1\sigma$ standard deviation, deduced from a Poissonian counting statistics of the single photon source. The errors in the experiments are mainly due to the imperfections of the single photon source, wave plates, and the photon detectors.  Our experimental setup has good scalability to carry out the generalized Deutsch's problem with $n>2$ Boolean functions, by setting the proper black-box operation $U_1$ and observing the responses from the port $a$ or the port $b$. 
The operation $U_1$ is defined as $U_1=D(f_1)\times D(f_2)\times \ldots \times D(f_n)$, where each $D(f_i)$ can be achieved through the wave plate. Each wave plate has four possible types of $\pm I$ and $\pm Z$. 
Therefore, there are total $4^n$ different possibilities for experimental realization of the $U_1$.


\emph{Conclusion.}---In conclusion, we proposed  a novel quantum algorithm with indefinite causal order to solve the generalized Deutsch's problem. Our algorithm not only reduces the number of queries to the black-box by half over the classical algorithm, but also for the first time outperforms the Deutsch's algorithm in terms of quantum gate number and circuit depth. We experimentally demonstrated the algorithm by 2-SWITCH in a stable Sagnac loop interferometer, which greatly simplifies the experiments and the duty cycle. The experimental results showed the both ultra-high and robust success probabilities.  Our experiments have shown a clear advantage of quantum circuit superposition without a fixed-gate order. We anticipate that the indefinite causal order quantum circuits may also have advantages in related problems, such as quantum Fourier transform.


\section*{ACKNOWLEDGMENTS}

This work is supported by  Beijing Institute of Technology Research Fund Program for Young Scholars, National Key Research and Development Program Earth Observation and Navigation Key Specialities (No. 2018YFB0504300).


\bibliography{mybibliography}
\bibliographystyle{apsrev4-2}

\clearpage
\onecolumngrid
\begin{center}
\textbf{\large Supplemental Material}
\end{center}
\setcounter{equation}{0}
\setcounter{figure}{0}
\setcounter{table}{0}
\renewcommand{\theequation}{S\arabic{equation}}
\renewcommand{\thefigure}{S\arabic{figure}}
\renewcommand{\thetable}{S\arabic{table}}

\section{A generalized Deutsch's algorithm} \label{AppendixA}

We generalize the Deutsch's algorithm to a generalized Deutsch's algorithm, which can solve the generalized Deutsch's problem with $n$ Boolean functions.  As shown in Fig. \ref{Deutsch}(b), the generalized Deutsch's algorithm begins with an initial state $|\psi_0\rangle=|0\rangle_1|1\rangle_2$ and after two Hardmard ($H$) gates are applied, the initial state  $|\psi_0\rangle$ becomes
\begin{equation}  \label{eqA1}
\begin{aligned}
|\psi_1\rangle=\frac{1}{2}\big(|0\rangle_1+|1\rangle_1\big)\otimes\big(|0\rangle_2-|1\rangle_2\big).
\end{aligned}
\end{equation}
Subsequently, the superposition state $|\psi_1\rangle$ is acted on $n$ black-box unitary transformations $U_{f_i}$ ($i$=1, 2, $\ldots$, $n$): $U_{f_i}|x\rangle|y\rangle=|x\rangle|y\oplus f_i(x)\rangle$, where $\oplus$ denotes addition modulo 2. These black-boxes evolves $|\psi_1\rangle$ as
\begin{equation}  \label{eqA2}
\begin{aligned}
|\psi_2\rangle=\pm\frac{1}{2}\big(|0\rangle_1+(-1)^{\bigoplus\limits_{i=1}^n (f_i(0)\oplus f_i(1))}|1\rangle_1\big)\otimes\big(|0\rangle_2-|1\rangle_2\big).
\end{aligned}
\end{equation}
Finally,  the first qubit is acted on a Hardmard gate again, obtaining
\begin{equation}  \label{eqA3}
|\psi_3\rangle=\left\{\begin{aligned}
&\pm|0\rangle_1\otimes\frac{|0\rangle_2-|1\rangle_2}{\sqrt{2}},  \;\; \bigoplus\limits_{i=1}^n f_i(0)=\bigoplus\limits_{i=1}^n f_i(1), \\
&\pm|1\rangle_1\otimes\frac{|0\rangle_2-|1\rangle_2}{\sqrt{2}},  \;\; \bigoplus\limits_{i=1}^n f_i(0)\neq\bigoplus\limits_{i=1}^n f_i(1).
\end{aligned}
\right.
\end{equation}
From Eq. (\ref{eqA3}), the generalized Deutsch's problem can be solved by measuring the first qubit: if $n$ is odd and the measurement result is in the state $|0\rangle_1$ ($|1\rangle_1$), which means there are (are not) an odd number of constant functions; if $n$ is even and the measurement result is in the state $|0\rangle_1$ ($|1\rangle_1$), which means there are not (are) an odd number of constant functions. The generalized Deutsch's algorithm only needs to query each $f_i(x)$ once to solve this problem, while one has to query each $f_i(x)$ twice classically.


\section{Our new algorithm to solve the generalized Deutsch's problem} \label{AppendixB}

We summary our proposed algorithm to solve the generalized Deutsch's problem and give a detail proof as follows.

\newtheorem*{theorem}{Algorithm}
\begin{theorem}  \label{Theo2} 
Let operators $D(f_i)$=$\sum\limits_{x\in\{0, 1\}}(-1)^{f_i(x)}|x\rangle\langle x|$ ($i$=1, 2, $\ldots, n$), $U_1=D(f_1)D(f_2)\ldots D(f_n)$ and $U_2=X$. $\forall y\in \{0,1\}$, (i) for  $n$ is odd,
\begin{eqnarray}              \label{eq4}
\begin{aligned}
&[U_1, U_2]|y\rangle=0  \Leftrightarrow    \text{$\exists$ an odd number of  constant functions}, \\
&\{U_1, U_2\}|y\rangle=0  \Leftrightarrow   \text{$\nexists$ an odd number of constant functions}.
\end{aligned}
\end{eqnarray}
(ii) For  $n$ is even,
\begin{eqnarray}              \label{eq5}
\begin{aligned}
    &[U_1, U_2]|y\rangle=0  \Leftrightarrow    \text{$\nexists$ 
 an odd number of   constant functions}, \\
&\{U_1, U_2\}|y\rangle=0  \Leftrightarrow   \text{$\exists$ an odd number of constant functions}.
\end{aligned}
\end{eqnarray}
Here [$U_1$, $U_2$]=$U_1U_2-U_2U_1$ and \{$U_1$, $U_2$\}=$U_1U_2+U_2U_1$ denote the commutator and anti-commutator of operators $U_1$ and $U_2$, respectively.
\end{theorem} 

\begin{proof}
We first apply $U_1$ on the state $|y\rangle$, obtaining $U_1|y\rangle=\prod\limits_{i=1}^{n}D(f_i)|y\rangle=(-1)^{\bigoplus\limits_{i=1}^n f_i(y)}|y\rangle$, and then $U_2$ is applied on the result, obtaining $U_2U_1|y\rangle=(-1)^{\bigoplus\limits_{i=1}^n f_i(y)}|y\oplus1\rangle$. Similarly, we can obtain $U_1U_2|y\rangle=(-1)^{\bigoplus\limits_{i=1}^n f(y\oplus1)}|y\oplus1\rangle$. For $n$ is odd, if $\bigoplus\limits_{i=1}^n f_i(y)=\bigoplus\limits_{i=1}^n f_i(y\oplus1)$, that means there are an odd number of constant functions, which  corresponds to the commutation of $U_1$ and $U_2$;  if $\bigoplus\limits_{i=1}^n f_i(y)\neq\bigoplus\limits_{i=1}^n f_i(y\oplus1)$, that means there are no odd number of constant functions, which  corresponds to the anti-commutation of $U_1$ and $U_2$. In the same way, it can be easily proved that the algorithm is valid for the case where $n$ is even.
\end{proof}
\qedhere 

\section{Experimental  realization of the gate operations} \label{AppendixC}

The experimental realization of gate operations $D(f_i)$ and $U_2$ is presented in Tab. \ref{table3}.

\begin{table} [htbp]
\centering\caption{The angles for the realization of gate operations $D(f_i)$ and $U_2$ in our algorithm by using two quarter-wave plates (QWPs) and one half-wave plate (HWP).}
\begin{tabular}{clcccccccc}

\hline  \hline
&Gate operation         &\quad   &\qquad\; QWP   &\quad   &\qquad\; HWP   &\quad   &\qquad\; QWP    \\

\hline

& $D(f_i)=I$    &\quad  &\qquad\; $0^\circ$      &\quad  &\qquad\;  $0^\circ$     &\quad  &\qquad\; $0^\circ$      \\
& $D(f_i)=-I$   &\quad  &\qquad\; $90^\circ$     &\quad  &\qquad\;  $0^\circ$     &\quad  &\qquad\; $90^\circ$     \\
& $D(f_i)=Z$    &\quad  &\qquad\; $0^\circ$      &\quad  &\qquad\;  $90^\circ$    &\quad  &\qquad\; $90^\circ$     \\
& $D(f_i)=-Z$   &\quad  &\qquad\; $90^\circ$     &\quad  &\qquad\;  $0^\circ$     &\quad  &\qquad\; $0^\circ$      \\
& $U_2=X$    &\quad  &\qquad\; $0^\circ$      &\quad  &\qquad\;  $45^\circ$    &\quad  &\qquad\; $0^\circ$      \\

\hline  \hline
\end{tabular}\label{table3}
\end{table}

\section{The experimental correspondence for our algorithm to solve the generalized Deutsch's problem} \label{AppendixD}

A detail list of our algorithm to solve the Deutsch's problem is shown in Tab. \ref{table1}.
\begin{table} [htb]
\centering\caption{An experimental correspondence for our algorithm to solve the Deutsch's problem.}
\begin{tabular}{cccccccccc}

\hline  \hline
&Function class       &\quad & Function $f$ &\quad & $U_1$ &\quad & $U_2$ &\quad  &Expect port   \\

\hline

& Constant  &\quad & $f(0)=f(1)=0$           &\quad &  $I$  &\quad & $X$ &\quad  &   Port $b$     \\
& Constant  &\quad & $f(0)=f(1)=1$           &\quad &  $-I$ &\quad & $X$ &\quad  &   Port $b$     \\
& Balanced  &\quad & $f(0)=0$ and $f(1)=1$   &\quad &  $Z$  &\quad & $X$ &\quad  &   Port $a$     \\
& Balanced  &\quad & $f(0)=1$ and $f(1)=0$   &\quad &  $-Z$ &\quad & $X$ &\quad  &   Port $a$     \\

\hline  \hline
\end{tabular}\label{table1}
\end{table}

A detail list of our algorithm to solve the generalized Deutsch's problem  with two Boolean functions is shown in Tab. \ref{table4}.
\begin{table} [htb]
\centering\caption{An experimental correspondence for our algorithm to solve the generalized Deutsch's problem with two Boolean functions $f_1$ and $f_2$.}
\begin{tabular}{cccccccccc}

\hline  \hline
& Function class        &\quad & Functions $f_1$ and $f_2$ &\quad & $U_1=D(f_1)\times D(f_2)$ &\quad & $U_2$ &\quad  &Expect port   \\

\hline

&Two constant $f_1$, $f_2$ &\quad & $f_1(0)=f_1(1)=0$, $f_2(0)=f_2(1)=0$           &\quad &  $I\times I$  &\quad & $X$ &\quad  &   Port $b$     \\
&Two constant $f_1$, $f_2$ &\quad & $f_1(0)=f_1(1)=0$, $f_2(0)=f_2(1)=1$           &\quad &  $I\times (-I)$  &\quad & $X$ &\quad  &   Port $b$     \\
&One constant $f_1$ &\quad & $f_1(0)=f_1(1)=0$, $f_2(0)=0$, $f_2(1)=1$           &\quad &  $I\times Z$  &\quad & $X$ &\quad  &   Port $a$     \\
&One constant $f_1$ &\quad & $f_1(0)=f_1(1)=0$, $f_2(0)=1$, $f_2(1)=0$           &\quad &  $I\times (-Z)$  &\quad & $X$ &\quad  &   Port $a$     \\
&Two constant $f_1$, $f_2$ &\quad & $f_1(0)=f_1(1)=1$, $f_2(0)=f_2(1)=0$           &\quad &  $-I\times I$  &\quad & $X$ &\quad  &   Port $b$     \\
&Two constant $f_1$, $f_2$ &\quad & $f_1(0)=f_1(1)=1$, $f_2(0)=f_2(1)=1$           &\quad &  $-I\times (-I)$  &\quad & $X$ &\quad  &   Port $b$     \\
&One constant $f_1$ &\quad & $f_1(0)=f_1(1)=1$, $f_2(0)=0$, $f_2(1)=1$           &\quad &  $-I\times Z$  &\quad & $X$ &\quad  &   Port $a$     \\
&One constant $f_1$ &\quad & $f_1(0)=f_1(1)=1$, $f_2(0)=1$, $f_2(1)=0$           &\quad &  $-I\times (-Z)$  &\quad & $X$ &\quad  &   Port $a$     \\
&One constant $f_2$ &\quad & $f_1(0)=0$, $f_1(1)=1$, $f_2(0)=f_2(1)=0$           &\quad &  $Z\times I$  &\quad & $X$ &\quad  &   Port $a$     \\
&One constant $f_2$ &\quad & $f_1(0)=0$, $f_1(1)=1$, $f_2(0)=f_2(1)=1$          &\quad &  $Z\times (-I)$  &\quad & $X$ &\quad  &   Port $a$     \\
&No constant $f$ &\quad & $f_1(0)=0$, $f_1(1)=1$, $f_2(0)=0$, $f_2(1)=1$         &\quad &  $Z\times Z$  &\quad & $X$ &\quad  &   Port $b$     \\
&No constant $f$ &\quad &$f_1(0)=0$, $f_1(1)=1$, $f_2(0)=1$, $f_2(1)=0$            &\quad &  $Z\times (-Z)$  &\quad & $X$ &\quad  &   Port $b$     \\
&One constant $f_2$ &\quad & $f_1(0)=1$, $f_1(1)=0$, $f_2(0)=f_2(1)=0$           &\quad &  $-Z\times I$  &\quad & $X$ &\quad  &   Port $a$     \\
&One constant $f_2$ &\quad & $f_1(0)=1$, $f_1(1)=0$, $f_2(0)=f_2(1)=1$          &\quad &  
 $-Z\times (-I)$  &\quad & $X$ &\quad  &   Port $a$     \\
&No constant $f$ &\quad & $f_1(0)=1$, $f_1(1)=0$, $f_2(0)=0$, $f_2(1)=1$         &\quad &  
 $-Z\times Z$  &\quad & $X$ &\quad  &   Port $b$     \\
&No constant $f$ &\quad &$f_1(0)=1$, $f_1(1)=0$, $f_2(0)=1$, $f_2(1)=0$            &\quad &  $-Z\times (-Z)$  &\quad & $X$ &\quad  &   Port $b$     \\

\hline  \hline
\end{tabular}\label{table4}
\end{table}


In experiments, one person sets up the wave plates to achieve the black-box operation $U_1$, and another one who does not know the type of $U_1$, but only knows the parity of $n$, can solve the generalized Deutsch’s problem by observing the responses from  the port $a$ or the port $b$.

%
%
%
%
%
%

\end{document}